\documentclass[a4paper]{article}

\usepackage[T2A]{fontenc}
\usepackage[utf8]{inputenc}

\usepackage{hyphenat}
\usepackage{amsthm}
\newtheorem{theorem}{Theorem}
\newtheorem{prop}{Proposition}

\newtheorem*{remark}{Remark}
\newtheorem{prop/defi}{Proposition/Definition}
\newtheorem{theor/defi}{Theorem/Definition}
\newtheorem{defi}{Definition}

\newtheorem{coro}{Corollary}

\usepackage[english]{babel}
\usepackage{graphicx}
\usepackage{amsmath}
\usepackage{amssymb}
\usepackage{cite}
\usepackage{cancel}
\usepackage[cmtip,all]{xy}
\usepackage{faktor}
\usepackage{comment}
\usepackage[dvipsnames]{xcolor}
\usepackage[titletoc]{appendix}
\usepackage{tikz}
\usepackage{subcaption}

\usepackage{listings}
\lstdefinelanguage{Sage}[]{Python}
{morekeywords={False,sage,True},sensitive=true}
\definecolor{dblackcolor}{rgb}{0.0,0.0,0.0}
\definecolor{dbluecolor}{rgb}{0.01,0.02,0.7}
\definecolor{dgreencolor}{rgb}{0.2,0.4,0.0}
\definecolor{dgraycolor}{rgb}{0.30,0.3,0.30}

\newcommand{\longsquiggly}{\xymatrix{{}\ar@{~>}[r]&{}}}

\graphicspath{ {images/} }

\newcommand {\Z}{\mathbb Z}

\def\f2{{\mathbb F}_{2}}
\def\fn{{\mathbb F}_{2^n}}
\newcommand{\ff}[1]{{\mathbb F}_{2}^{#1}}
\newcommand{\FF}[1]{{\mathbb F}_{2^{#1}}}

\title{On Dillon's property of $(n,m)$-functions}
\author{Matteo Abbondati$^{1,2}$, Marco Calderini$^2$, Irene Villa$^2$}
\date{}

\begin{document}

\maketitle
\begin{center}
	$^1$LIRMM, University of Montpellier, CNRS, Montpellier, France.\\
	$^{2}$ University of Trento, Via Sommarive, 14, 38123 Povo, Trento, Italy.\\ matteo.abbondati@umontpellier.fr,$\{$marco.calderini,irene.villa$\}$@unitn.it
\end{center}

\begin{abstract}
    Dillon observed that an APN function $F$ over $\ff{n}$ with $n$ greater than $2$ must satisfy 
the condition $\{F(x) + F(y) + F(z) + F(x + y + z) \,:\, x,y,z \in\ff n\}= \ff n$. Recently, Taniguchi (2023) generalized this condition to functions defined from $\ff n$ to $\ff m$, with $m>n$, calling it the D-property.
Taniguchi gave some characterizations of APN functions satisfying the D-property and provided some families of APN functions from $\ff n$ to $\ff{n+1}$ satisfying this property.
In this work, we further study the D-property for  $(n,m)$-functions with $m\ge n$. We give some combinatorial bounds on the dimension $m$ for the existence of such functions. Then, we characterize the D-property in terms of the Walsh transform and for quadratic functions we give a characterization of this property in terms of the ANF. We also give a simplification on checking the D-property for quadratic functions, which permits to extend some of the APN families provided by Taniguchi. We further focus on the class of the plateaued functions, providing conditions for the D-property.
\end{abstract}

\section{Introduction}
A vectorial Boolean function, or $(n,m)$-function, is a function $F:\ff n\to \ff m$, where $\ff n$ denotes the $n$ dimensional vector space over the binary field $\ff{}$. 
Vectorial Boolean functions play an important role in many different areas of mathematics and computer science.

In cryptography and, in particular, in the design of block ciphers, vectorial Boolean functions are of critical importance. Virtually, all modern block ciphers incorporate $(n,m)$-functions (usually called \emph{S-boxes}) which represent the only nonlinear components involved in the operations of the cipher, and as such, the security of the encryption directly depends on the properties of these functions. 

Among the most efficient attacks on block ciphers we have the differential cryptanalysis \cite{BS90}, based on the study of how differences in input can affect the resultant differences in the output. 

The resistance to differential attacks for an $(n,m)$-function $F$ is high when the value 
$$\delta_F = \max_{a\ne0\in \ff n,b\in\ff m} |\{x\in\ff n\,:\, F(x+a)+F(x)=b\}|$$
is small.

When $n \le m$, the differential uniformity of an $(n, m)$-function is at least $2$, and functions meeting this bound are called almost perfect nonlinear (APN). 

Usually, APN functions have been studied for the case $m=n$, see \cite{carlet,CBC21}  for known APN functions. 
However, studying vectorial Boolean functions with a number of output bits greater than the number of input bits is also of interest. In \cite{C23}, Carlet studied $(n,n+1)$-functions for constructing balanced functions with good nonlinearity.
Also in \cite{ACCSZ19}, the authors studied some cryptographic properties of certain Feistel networks involving $(n,m)$-functions, with $n<m$, as S-boxes. In particular, in \cite{ACCSZ19} it is given an example of a cipher using an APN function from $4$ bits to $5$ bits as an S-box.

 Recently,  Taniguchi in \cite{taniguchi} focused on the study of APN $(n,n+1)$-functions as well. 
His work was based on an observation noted by Dillon,  as
Carlet reported in his book \cite[pg. 381]{carlet},  that is, an APN function $F:\ff n\to \ff n$ must satisfy the following condition:
$$
F(x)+F(y)+F(z)+F(x+y+z)=c
$$
admits a solution $(x,y,z)\in (\ff{n})^3$ for any $c\in\ff n$.
Taniguchi in \cite{taniguchi} referred to this condition as the \emph{D-property}, and he studied APN $(n,m)$-functions, with $m>n$, which satisfy the D-property.
In particular, the author focused on the case of $(n,n+1)$-functions and obtained several classes of APN functions defined from $\ff n$ to $\ff{n+1}$ having the D-property.

In this work, we further study  $(n,m)$-functions  satisfying the D-property.
In the first part, we give an upper bound on the dimension $m$ for the existence of these APN functions. We also characterize general $(n,m)$-functions satisfying the D-property in terms of their difference distribution table and in terms of their fourth Walsh moment.
Then, we focus on quadratic  $(n,m)$-functions. We provide a characterization of the D-property for these functions in terms of their Walsh transform and in terms of their algebraic normal form. Moreover, we give a simplification for checking the D-property for quadratic  functions. This permits to extend some of the APN families provided by Taniguchi.
Afterwards, we  characterize plateaued vectorial Boolean functions having the D-property and we study a particular case of strongly plateaued APN functions.
In the last part, we describe a connection of some results obtained with the higher-order differentiability \cite{LAI}.

\section{Preliminaries}\label{sec:2}
We recall some standard notation. Given a set $A$, with $|A|$ we denote its size and with $A^*$ we denote $A\setminus\{0\}$. With $\Z$ we identify the ring of integers and with  $\f2$  the finite field with 2 elements.
For $n$ a positive integer, we denote with $\ff n$ the $n$-dimensional vector space over $\f2$ and with $\FF n$ the finite field with $2^n$ elements. 

In general, when not specified, $n$ and $m$ are two positive integers. An $(n,m)$-function $F$ is a map from $\ff n$ to $\ff m$. If $F(0)=0$ we say that the function is \emph{normalized}.
For a nonzero $v\in\ff m$,  the $(n,1)$-function $F_v=v\cdot F$ is called a \emph{component} of $F$, where ``$\cdot$'' is a chosen inner product in $\ff m$.
The \emph{algebraic normal form} of $F$ (shortly ANF) is its (unique) representation as an $n$-variable polynomial 
$$F(x_1,\ldots,x_n)=\sum_{I\subseteq\{1,\ldots,n\}}a_I\prod_{i\in I}x_i\in\ff m[x_1,\ldots,x_n]/(x_1^2+x_1,\ldots,x_n^2+x_n).$$ Sometimes we  identify the set of variables $(x_1,\ldots,x_n)$ simply with $x$.

Given the isomorphism between $\ff n$ and $\FF n$,
 $(n,m)$-functions can also be viewed as taking their inputs in $\FF n$ and, if $m$ divides $n$, the output can be expressed as a polynomial function of the input. 
In general, an $(n,n)$-function $F$ can be also represented as a polynomial of the form
$$F(X)=\sum_{i=0}^{2^n-1}b_iX^i\in\fn[X].$$ This is called the \emph{univariate representation} of $F$, and it is unique.

In this work, we make use of both representations. It will be clear from the context which one is under consideration. 

The \emph{algebraic degree} of $F$ ($\deg(F)$) is equivalently the largest binary weight of an exponent $i$ with $b_i\ne0$ for the univariate representation, or the largest size of a set $I$ with $a_I\ne0$ for the ANF.
Then,  $F$ is called \emph{quadratic} if $\deg(F)\le2$, \emph{affine} if $\deg(F)\le1$ and \emph{linear} if normalized with $\deg(F)\le1$.
For $m|n$, a well known example of linear $(n,m)$-function is the \emph{trace} map:
$$Tr^n_m(X)=\sum_{i=0}^{n/m-1}X^{2^{im}}.$$
To simplify the notation, when $m=1$ we denote it as $Tr_n(X)=\sum_{i=0}^{n-1}X^{2^i}$, often called the \emph{absolute trace function}.
Note that, given the isomorphism between $\ff n$ and $\FF n$, the trace map $Tr^n_m$ can be also viewed as a map from $\ff n$ to $\ff m$.

In the following, we  present some well known results and definitions on vectorial Boolean functions considered as maps over binary vector spaces. We remind to the reader that these notions can be translated to the finite field case. We refer the interested reader to \cite{carlet} for more results on vectorial Boolean functions.

Given an $(n,m)$-function $F$ and a nonzero element $a\in\ff n$, the \emph{derivative of $F$ in the direction $a$} is the map $D_aF:\ff n\rightarrow\ff m$ defined as $D_aF(x)=F(x+a)+F(x)$.
It is possible to iterate the derivative of $F$ obtaining the \emph{second order derivative} of $F$ in the directions $a,a'$  as follows:
$$D_{a,a'}^2F(x)=D_{a'}\left(D_aF(x)\right)=F(x+a+a')+F(x+a')+F(x+a)+F(x).$$ 
 Therefore, the second order derivative satisfies the following relation 
 \begin{equation}\label{eq:der2}
    D_{a,a'}^2F(x)=D_aF(x)+D_aF(x+a').
\end{equation}

The \emph{difference distribution table} of $F$ ($\texttt{DDT}_F$) is the $(2^n-1)\times2^m$ table, labeled by elements  $(a,b)\in\ff n\setminus\{0\}\times\ff m$, defined as
$$\texttt{DDT}_F(a,b)=|\{x\in\ff n : D_aF(x)=b\}|.$$
The \emph{differential uniformity} of $F$ ($\delta_F$) is then defined as the maximal value in $\texttt{DDT}_F$. A function with differential uniformity $\delta$ is said to be \emph{$\delta$-uniform}.
All the entries of the \texttt{DDT}, hence the differential uniformity, are necessarily even.
Nyberg's bound \cite{Nyberg}  tells us that $\delta_F\ge2^{n-m}$, so for $(n,n)$-functions the lowest possible value for the differential uniformity is 2. Functions attaining this bound are called \emph{almost perfect nonlinear} (APN).

Another important definition for vectorial Boolean functions is the \emph{Walsh transform}.
Assume that an inner product in $\ff n$ and an inner product in $\ff m$ have been chosen, both denoted by ``$\cdot$''.
The \emph{Walsh transform} of an $(n,m)$-function $F$ ($\mathcal W_F$) maps any pair $(u,v)\in\ff n\times\ff m$  to the following element in $\Z$,
$$\mathcal W_F(u,v)=\sum_{x\in\ff n}(-1)^{v\cdot F(x)+u\cdot x}.$$
The \emph{Walsh spectrum} of $F$ is then the multi-set of all the values $\mathcal W_F(u,v)$, for $u\in\ff n$ and $v\in\ff m$.
The \emph{nonlinearity} of $F$, defined as the minimum (Hamming) distance between its (nonzero) components $v\cdot F$ and affine functions, can be proven to be equal to
$$nl(F)= 2^{n-1}-\frac{1}{2}\max_{\substack{u\in\ff n\\v\in\ff m\setminus\{0\}}}|\mathcal W_F(u,v)|.$$
An $(n,m)$-function $F$ is \emph{plateaued} if for any $v\in\ff m\setminus\{0\}$ there exists $\lambda_v$ such that $\mathcal W_F(u,v)\in\{0,\pm\lambda_v\}$ for every $u\in\ff n$.

We recall some equivalence relations used for $(n,m)$-functions. Consider two such functions $F$ and $F'$.
\begin{itemize}
    \item If there exist two affine permutations $A_1$ and $A_2$ of $\ff m$ and $\ff n$ respectively, such that $F'(x)=A_1\circ F\circ A_2(x)$, then $F$ and $F'$ are called \emph{affine equivalent}.
    \item If there exist two  $(n,m)$-functions $A$ and $F''$ with $A$ affine and $F''$ affine equivalent to $F$, such that $F'(x)=F''(x)+A(x)$, then $F$ and $F'$ are called \emph{extended affine equivalent}, shortly EA-equivalent.
    \item If there exists an affine permutation $\mathcal{A}$ over $\ff n\times\ff m$ such that $\mathcal A(\Gamma_{F})=\Gamma_{F'}$, where $\Gamma_F=\{(x,F(x)) : x\in\ff n\}$ and $\Gamma_{F'}=\{(x,F'(x)) : x\in\ff n\}$ are the graphs of $F$ and $F'$, then $F$ and $F'$ are called \emph{CCZ-equivalent} \cite{CCZ}. 
\end{itemize}
We also recall  that the affine equivalence is a particular case of the EA-equivalence, which is itself a particular case of the CCZ-equivalence. Moreover, the differential uniformity is kept invariant by these equivalence relations.

\subsection{Dillon's property}
We can now introduce the core definition of this work, Dillon's property.
Related to the study of nonlinearity of  APN functions, Dillon observed the following, as reported  in \cite[pg. 381]{carlet}.
Given $F$ a (normalized) APN $(n,n)$-function with $n>2$, then for any nonzero $c\in\ff n$ the equation 
$$F(x)+F(y)+F(z)+F(x+y+z)=c$$ must have a solution $(x,y,z)\in(\ff n)^3$.
This means that, for $F$ APN, the following two conditions related to the second order derivative are satisfied.\begin{itemize}
    \item For any distinct nonzero $a,a'\in\ff n$, the second order derivative $D_{a,a'}^2F$ takes only nonzero values.
    \item As $(a,a',x)$ varies inside $(\ff n)^*\times(\ff n)^*\times\ff n$, the evaluations $D_{a,a'}^2F(x)$ take every possible nonzero value $c\in\ff m$.
\end{itemize}

Later on, Taniguchi in \cite{taniguchi} generalized the above property introducing the following class of functions.
\begin{defi}
    An $(n,m)$-function $F$ has the Dillon's property (D-property) if 
    $$\left\{F(x)+F(y)+F(z)+F(x+y+z) : x,y,z\in\ff n\right\}=\ff m.$$
    We call such function a D-function.
\end{defi}
Hence, Dillon's observation corresponds to the following result.
\begin{prop}
	\label{Dillon property}
	Let $F$ be an $(n,n)$-function with $n>2$. If $F$ is APN then it is a D-function.
\end{prop}

In \cite{taniguchi} the author studied the notion of D-function for APN $(n,m)$-functions with $n<m$, particularly to the case $m=n+1$ trying to characterize for which $(n,n+1)$-functions the D-property holds.
Indeed,  unlike for the case $n=m$, in general it is not true that if an $(n,m)$-function $F$ is APN then it  satisfies the D-property.
\begin{remark}
    This is easily verified by considering the Gold APN $(7,7)$-function $F(X)=X^3$, which clearly satisfies the D-property.
    When restricting $F$ to the hyperplane $T_0$ of trace zero elements, we obtain a $(6,7)$-function $\left.F\right|_{T_0}$. This function  is still an APN function, since it is the restriction of an APN map, but it does not satisfy the D-property.
    Indeed via a computer search we can verify that $e_1\not\in\{D_{a,a'}^2F(\omega) : a,a'\in T_0^*, \omega\in T_0\}$, where $e_1=(1,0,\ldots,0)$.
\end{remark}

Taniguchi in \cite{taniguchi} presented the following results on D-functions.
\begin{prop}
	\label{Prop1 Taniguchi}
	Let $F$ be a normalized $(n,m)$-function.
	Then $F$ has the D-property if and only if $\{D^2_{a_1,b}F(0)+D^2_{a_2,b}F(0) : a_1,a_2,b\in\mathbb{F}_2^n\}=\ff m$.
 
	If furthermore $F$ is quadratic, then it has the D-property if and only if $\{D^2_{a,b}F(0) : a,b\in\mathbb{F}_2^n\}=\mathbb{F}_2^m$.
\end{prop}
\begin{prop}
	\label{Prop3 Taniguchi}
	If $F$ is a normalized $(n,m)$ D-function, then $nl(F)>0$.
\end{prop}

Next in \cite{taniguchi} the author built several families of quadratic APN functions $F:\ff n\rightarrow\ff{n+1}$ satisfying the D-property.
We adopt the notation used by the author, describing the functions in their univariate representation.
Set $T_0(\FF k)=\{\omega\in\FF k : Tr_k(\omega)=0\}$.
\begin{theorem}
	Let $n+1=2k$ with $k\geqslant 3$. 
 Let $F:\mathbb{F}_{2^{n+1}}\rightarrow\mathbb{F}_{2^{n+1}}$ be a normalized quadratic APN function with univariate representation given by
	\begin{equation*}
		F(X)=\sum_{i<j}a_{i,j}X^{2^i+2^j}+\sum_{\ell}b_\ell X^{2^\ell}
	\end{equation*}
	 with  all the coefficients $a_{i,j},b_\ell$ belonging to $\FF k$ and satisfying $i+j\equiv 1\bmod2$ for every $i,j$ with $a_{i,j}\ne 0$.  
	 Then $\left.F\right|_{T_0(\mathbb{F}_{2^{n+1}})}:T_0(\mathbb{F}_{2^{n+1}})\rightarrow\mathbb{F}_{2^{n+1}}$ satisfies the D-property.
\end{theorem}
As a corollary, the author showed that the Gold functions $X^{2^i+1}$ restricted to the subspace of trace-zero element of $\mathbb{F}_2^{n+1}$ satisfies the D-property for every $i$ odd and for every dimension $n+1$ even ($n+1\ge6$). 

For the case $n+1$ odd, the following result was obtained in \cite{taniguchi}.
\begin{theorem}
	\label{Theor n+1 odd taniguchi}
	Let $F:\mathbb{F}_{2^{n+1}}\rightarrow\mathbb{F}_{2^{n+1}}$ be a quadratic normalized APN function such that all the coefficients of the univariate representation of $F$ belong to a subfield $\FF k$ of $\FF{n+1}$, with $\FF k\ne \f2,\mathbb{F}_4$.
	Then if $\left.F\right|_{T_0(\FF k)}:T_0(\FF k)\rightarrow\FF k$ satisfies the D-property, also $\left.F\right|_{T_0(\mathbb{F}_{2^{n+1}})}:T_0(\mathbb{F}_{2^{n+1}})\rightarrow\mathbb{F}_{2^{n+1}}$ satisfies the D-property.
\end{theorem}
This result was used to generate families of D-functions. In particular, in \cite{taniguchi} the author tested computationally the D-property for the Gold functions $X^{2^i+1}$ on $\mathbb{F}_{2^{n+1}}$ when restricting them to  $T_0(\mathbb{F}_{2^{n+1}})$ for $n+1=9,11,13,15$.
Thanks to Theorem \ref{Theor n+1 odd taniguchi}  the D-property could be extended to  the Gold functions in all the multiple dimensions, i.e.\ $n+1 = 9s,11s,13s,15s$ for every positive integer $s$. The same result was obtained for the APN function $X^3+Tr_{n+1}(X^9)$ restricted to  $T_0(\mathbb{F}_{2^{n+1}})$.

\section{Further characterizations of Dillon's property}
In this section we present our results on Dillon's property without assuming any particular structure on the function studied. In later sections we  focus more on the case of quadratic functions.\\
In what follows $F$ and $F'$ will always be an $(n,m)$-functions.

In \cite{taniguchi}, the author showed that, when restricted to quadratic functions, the D-property is invariant under the CCZ-equivalence.
We further investigate the invariance of the D-property for generic $(n,m)$-functions.

First, we assume that  there exist two affine permutations $A_1$ and $A_2$ of $\ff m$ and $\ff n$ respectively, such that $F'(x)=A_1\circ F \circ A_2(x)$.
We can write the first affine permutation as $A_1(x)=L_1(x)+c_1$, with $L_1$ a linear permutation of $\ff m$ and $c_1\in\ff m$.
We then want to study the equation $b=F'(x)+F'(y)+F'(z)+F'(x+y+z)$.
The right-hand-side coincides with the following
$$\begin{aligned}
F'(x)+F'(y)+F'(z)+&F'(x+y+z)=A_1(F(A_2(x)))+A_1(F(A_2(y)))\\&+A_1(F(A_2(z)))+A_1(F(A_2(x+y+z)))\\
=&L_1(F(A_2(x)))+L_1(F(A_2(y)))\\&+L_1(F(A_2(z)))+L_1(F(A_2(x+y+z))).
\end{aligned}$$

Given that $A_2(x+y+z)=A_2(x)+A_2(y)+A_2(z)$, by replacing $x'=A_2(x)$, $y'=A_2(y)$ and $z'=A_2(z)$ we obtain $$L_1^{-1}(b)=F('x)+F(y')+F(z')+F(x'+y'+z').$$
From the above computation, it is straightforward to deduce that the D-property is invariant under the affine transformation.
Now we set $F'(x)=F(x)+A(x)$, with $A$ an affine $(n,m)$-function.
Then the equation we want to study is of the  form
$b=F'(x)+F'(y)+F'(z)+F'(x+y+z)=F(x)+A(x)+F(y)+A(y)+F(z)+A(z)+F(x+y+z)+A(x+y+z)=F(x)+F(y)+F(z)+F(x+y+z)$.
Therefore, the D-property is invariant also when adding an affine transformation.

The two obtained results lead to the invariance with respect to the extended affine equivalence.
So, we can state the following.
\begin{prop}
    \label{prop: EA invariance}
    The D-property is invariant under EA-equivalence.
\end{prop}

Therefore, in the following sections, when studying the D-property of a function, we can assume without loss of generality that $F(0)=0$.

\subsection{A Combinatorial-Geometric Description}
\label{Sec.3.2:A Combinatorial-Geometric Description}

Like the APN property, also the D-property can be described in terms of the values taken by the second order derivatives of the function.
Consider the set of all 2-dimensional affine subspaces of $\ff{n}$, $$\mathcal{AFF}_{2,n}=\big\{\  \{x,y,z,x+y+z\} : x,y,z\in\ff n \mbox{ s.t. } x\ne y, x\ne z, y\ne z\big\}.$$
Then, to an $(n,m)$-function $F$ we associate the map
\begin{align*}
	\Phi_F:&\mathcal{AFF}_{2,n}\longrightarrow\mathbb{F}_2^m\\
	&\hspace{2pc}A\longmapsto\sum_{x\in A}F(x).
\end{align*}
We see in the following that, as  the APN property, the D-property
 has a natural combinatorial-geometric description in terms of the introduced map.
Indeed, for $A=\{x,y,z,x+y+z\}\in \mathcal{AFF}_{2,n}$ we have $\Phi_F(A)=D_{x+y,x+z}^2F(x)$.
Therefore we have the following.
\begin{prop}
    An  $(n,m)$-function $F$ is APN if and only if  $\Phi_F(A)\ne 0$ for every $A\in\mathcal{AFF}_{2,n}$.
\end{prop}

\begin{prop}
    \label{prop.combinat.descript.}
	An  $(n,m)$-function $F$ is a D-function if and only if $\mathbb{F}_2^{m}\setminus\{0\}\subseteq\{\Phi_F(A) : A\in\mathcal{AFF}_{2,n}\}$.  
\end{prop}
This last description allows to derive necessary conditions on the dimensions $n$ and $m$ when there exists an APN $(n,m)$-function satisfying the D-property.
\begin{prop}
	\label{prop4}
	Let $n>2$ and let  $F$ be an APN $(n,m)$-function satisfying the D-property. Then the following bounds must hold
	\begin{equation*}
		n\leqslant m\leqslant\log_2\left(\frac{2^n}{12}(2^n-1)(2^{n-1}-1)+1\right)<3n-4.
	\end{equation*}
\end{prop}
\begin{proof}
The first inequality follows from Nyberg's bound since $F$ is APN. 
Indeed, we need $m\ge n-1$ and $m=n-1$ is impossible except if $n\le2$. This is due to the fact that  2-uniform $(n,n-1)$-functions would attain Nyberg's bound and, from \cite{Nyberg}, this would imply $n-1\le n/2$.

For the second bound we use Proposition \ref{prop.combinat.descript.} to obtain that if $F$ is a D-function then it must be
\begin{equation*}
	|\mathbb{F}_2^m\setminus\{0\}|=2^m-1\leqslant|\mathcal{AFF}_{2,n}|.
\end{equation*}
It is known that $|\mathcal{AFF}_{2,n}|=2^{n-2}\binom{n}{2}_2$ where 
\begin{equation*}
	\binom{n}{k}_q=\frac{(1-q^n)(1-q^{n-1})\cdots(1-q^{n-k+1})}{(1-q)(1-q^2)\cdots(1-q^k)}
\end{equation*}
is the Gaussian binomial coefficient which counts the number of vector subspaces of dimension $k$ inside $\mathbb{F}_q^n$.\\
In our case we get
\begin{equation*}
	|\mathcal{AFF}_{2,n}|=2^{n-2}\frac{(1-2^n)(1-2^{n-1})}{(1-2)(1-2^2)}=\frac{1}{12}2^n(2^n-1)(2^{n-1}-1).
\end{equation*}
Hence we can conclude that
\begin{equation*}
	2^m\leqslant\frac{2^n}{12}(2^n-1)(2^{n-1}-1)+1,
\end{equation*}
which leads to the upper bound on $m$ by taking logarithms both sides.\\
We can make this more explicit; since $n>2$ the bound above is surely odd, hence we can drop the $+1$ and get
\begin{equation*}
	2^m\leqslant\frac{2^n}{12}(2^n-1)(2^{n-1}-1).
\end{equation*}
By taking logarithms on both sides, we obtain the last bound
\begin{equation*}
	m\leqslant n+\log_2(2^n-1)+\log_2(2^{n-1}-1)-\log_2(12)<3n-1-3=3n-4.
\end{equation*}
\end{proof}

\subsection{Dillon's Property from the \texttt{DDT}}
\label{Sec.3.3:Dillon's Property from the DDT}

It is possible to check the Dillon's property for a general function (not necessarily APN) $F:\mathbb{F}_2^n\longrightarrow\mathbb{F}_2^m$, from its difference distribution table. The key relation exploited here is the one expressing the second order derivatives as a combination of first order derivatives, see Equation \eqref{eq:der2}.

From this relation, it follows that the values of the second order derivatives can be read from the \texttt{DDT} as 2-weight sums of the output differences corresponding to two nonzero entries on the same row of the difference distribution table. More precisely we have the following result.
\begin{prop}
	Let $F$ be an $(n,m)$-function. Then for every nonzero $\alpha\in\mathbb{F}_2^n$ we have that
	\begin{equation*}
		\{D^2_{\alpha,\beta}F(w) : \beta,w\in\mathbb{F}_2^n\}=\{b_1+b_2 : \emph{\texttt{DDT}}_F(\alpha,b_1)\ne 0,\emph{\texttt{DDT}}_F(\alpha,b_2)\ne 0\}.
	\end{equation*}
\end{prop} 
\begin{proof}
    We are going to show this equality by showing the double inclusion of the sets involved.

Let $y=D^2_{\alpha,\beta}F(w)$ for some $\beta,w\in\mathbb{F}_2^n$. Thanks to Equation \eqref{eq:der2} we see that
$y=D_{\alpha}F(w)+D_{\alpha}F(w+\beta)$.
By construction of the \texttt{DDT} we see that if setting $b_1=D_{\alpha}F(w)$ and $b_2=D_{\alpha}F(w+\beta)$ we have $\texttt{DDT}_F(\alpha,b_1),\texttt{DDT}_F(\alpha,b_2)\ne0$.
Therefore we have\begin{equation*}
	y\in\{b_1+b_2 : \texttt{DDT}_F(\alpha,b_1)\ne 0,\texttt{DDT}_F(\alpha,b_2)\ne 0\}.
\end{equation*}
We show now the other inclusion.
Let $y=b_1+b_2$ with $\texttt{DDT}_F(\alpha,b_1),\texttt{DDT}_F(\alpha,b_2)\ne 0$. By definition of  \texttt{DDT}, we get that there exist $w_1,w_2\in\mathbb{F}_2^n$ such that $D_{\alpha}F(w_1)=b_1$ and $D_{\alpha}F(w_2)=b_2$.
So $y=D_{\alpha}F(w_1)+D_{\alpha}F(w_2)=D^2_{\alpha,w_1+w_2}F(w_1)$,
where in the last equality we used Equation \eqref{eq:der2}, and we conclude that
\begin{equation*}
	y\in\{D^2_{\alpha,\beta}F(w) : \beta,w\in\mathbb{F}_2^n\}.
\end{equation*}
\end{proof}

Thanks to the above proposition we obtain the following equivalent characterization of the D-property in terms of the \texttt{DDT} table of the function $F$.
\begin{theorem}
	\label{Theor Dillon from DDT}
	Let $F$ be an $(n,m)$-function. Then $F$ satisfies the D-property if and only if
	\begin{equation*}
		\bigcup_{\alpha\in\mathbb{F}_2^n\setminus\{0\}}\{b_1+b_2 : \emph{\texttt{DDT}}_F(\alpha,b_1)\ne 0, \emph{\texttt{DDT}}_F(\alpha,b_2)\ne 0\}=\mathbb{F}_2^m.
	\end{equation*}
\end{theorem}

For every $\alpha\in\mathbb{F}_2^n\setminus\{0\}$
we consider the  $m\times(2^m-1)$ matrix $\mathcal{D}_{F,\alpha}$ built by columns such that for every $b\in\mathbb{F}_2^m\setminus\{0\}$, the corresponding column is given by:
\begin{equation*}
	\left[\mathcal{D}_{F,\alpha}\right]_{b}=\begin{cases}
		b&\text{ if }\texttt{DDT}_F(\alpha,b)\ne 0\\
		0&\text{ otherwise.}
	\end{cases}
\end{equation*}
We label the nonzero columns of this matrix, say $b^{\alpha}_1,\ldots,b^{\alpha}_{k_{\alpha}}$, and we take the sub-matrix:
\begin{equation*}
	\mathbf{D}_{F,\alpha}=\left(
	\begin{array}{cccc}
		\vdots & \vdots & \vdots & \vdots\\
		b^{\alpha}_1& b^{\alpha}_2 & \ldots &b^{\alpha}_{k_{\alpha}}\\
		\vdots & \vdots & \vdots & \vdots\\
	\end{array}
	\right).
\end{equation*}
Then Theorem \ref{Theor Dillon from DDT} can be restated as follows.
\begin{theorem}
	Let $F$ be an $(n,m)$-function. Then $F$ satisfies the D-property if and only if
	\begin{equation*}
		\bigcup_{\alpha\in\mathbb{F}_2^n\setminus\{0\}}\mathbf{D}_{F,\alpha}(\mathcal H_{\alpha,2})=\mathbb{F}_2^m
	\end{equation*}
	where $\mathcal H_{\alpha,2}$
	is the set of vectors in $\mathbb{F}_2^{k_\alpha}$ whose Hamming weight is $2$,
	\begin{equation*}
		\mathbf{D}_{F,\alpha}(\mathcal H_{\alpha,2})=\{\mathbf{D}_{F,\alpha}\cdot v^T : v\in \mathcal{H}_{\alpha,2}\}
	\end{equation*}
and $\cdot$ is the usual matrix product.
\end{theorem}

\subsection{Relation with the Fourth Moment}

We can also characterize the D-property by means of the fourth moment of the Walsh transform.

\begin{theorem}\label{thm:4th moment}
Let $m\geqslant n$ and let $F:\mathbb{F}_2^{n}\rightarrow\mathbb{F}_2^{m}$.
	Then $F$ satisfies the D-property if and only if
	\begin{equation*}
		\sum_{(u,v)\in \mathbb{F}_2^n\times\mathbb{F}_2^m}(-1)^{v\cdot b}\mathcal W_{F}^4(u,v)>0
	\end{equation*} 
	for every $b\in\mathbb{F}_2^m\setminus\{0\}$.
\end{theorem}
\begin{proof}
 We have   $$
 \begin{aligned}
    &\sum_{(u,v)\in \mathbb{F}_2^n\times\mathbb{F}_2^m}(-1)^{v\cdot b}\mathcal W_{F}^4(u,v)\\
    &\quad =\sum_{v \in \ff m}(-1)^{v\cdot b}\sum_{x,y,z,t\in \ff n}(-1)^{v\cdot(F(x)+F(y)+F(z)+F(t))}\sum_{u\in\ff n}(-1)^{u\cdot(x+y+z+t)}\\
    &\quad =2^n\sum_{x,y,z\in \ff n}\sum_{v \in \ff m}(-1)^{v\cdot(F(x)+F(y)+F(z)+F(x+y+z)+b)}\\
    &\quad =2^{n+m}\cdot |\{(x,y,z)\in (\ff n)^3 : F(x)+F(y)+F(z)+F(x+y+z)=b\}|.
    \end{aligned}
    $$
    This concludes the proof.
\end{proof}

 We can obtain an alternative characterization from the above theorem if we introduce the following pseudo-Boolean function associated to the function $F$
\begin{align*}
	&\hspace{-1pc}H_{4,F}:\mathbb{F}_2^m\longrightarrow\mathbb{Z}\\
	&\hspace{2pc}v\longmapsto\sum_{u\in \mathbb{F}_2^n}\mathcal W_F^4(u,v).
\end{align*}
First of all, we recall the definition of \emph{Fourier-Hadamard transform} for a pseudo-Boolean function $\phi:\ff n\rightarrow\Z$
$$\mathcal{F}(\phi)(w)=\sum_{x\in\ff n}\phi(x)(-1)^{u\cdot x}.$$

From Theorem \ref{thm:4th moment}, we can express the D-property of an $(n,m)$-function in terms of the Fourier-Hadamard transform of the pseudo-Boolean function $H_{4,F}$.
\begin{prop}
	Let $m\geqslant n$ and let $F:\mathbb{F}_2^{n}\longrightarrow\mathbb{F}_2^{m}$.
	Then $F$ satisfies the D-property if and only if
	\begin{equation*}
		\mathcal{F}\left(H_{4,F}\right)(b)>0\hspace{2pc}\forall b\in\mathbb{F}_2^m.
	\end{equation*}
\end{prop}

\section{The Case of Quadratic Functions}
\label{Sec.3.4:The Case of Quadratic Functions}
Since the Dillon's property depends on the behavior of the second order derivatives of the function studied, it is no surprising that the case of quadratic functions will allow considerable simplifications.

We recall to the reader that the second order derivatives of quadratic functions are constant maps.
Indeed, given $F$ a quadratic $(n,m)$-function, it can be easily shown that
 $D^2_{\alpha,\beta}F(x)=F(0)+F(\alpha)+F(\beta)+F(\alpha+\beta)$ depends only on the directions $\alpha$ and $\beta$ and not on the point $x$.
From this observation follows that, for quadratic functions, to study the image set of the map $\Phi_F$ introduced at the beginning of Section \ref{Sec.3.2:A Combinatorial-Geometric Description}, it is enough to study its restriction over the set of $2$-dimensional vector subspaces in $\mathbb{F}_2^n$, that is
\begin{equation*}
	\mathcal{V}_{2,n}=\big\{\ \{0,x,y,x+y\} : x,y\in\ff n\setminus\{0\}\mbox{ s.t.\ } x\ne y \big\}\subset\mathcal{AFF}_{2,n}.
\end{equation*}
This allows to refine the bound in Proposition \ref{prop4} for the case of quadratic APN functions.
\begin{prop}
	\label{prop8}
	Let $F$ be an APN quadratic $(n,m)$-function with $n>3$. Then, if $F$ satisfies also the D-property, the following bounds must hold
	\begin{equation*}
		n\leqslant m< 2n-2.
	\end{equation*}
\end{prop}
\begin{proof}
    The proof goes exactly as for the one of Proposition \ref{prop4} with the only difference that instead of $\mathcal{AFF}_{2,n}$, the bound  is determined by $\mathcal{V}_{2,n}$, noticing that
\begin{equation*}
	|\mathcal{V}_{2,n}|=\binom{n}{2}_2=\frac{(1-2^n)(1-2^{n-1})}{(1-2)(1-2^2)}=\frac{(2^n-1)(2^{n-1}-1)}{3}.
\end{equation*}
\end{proof}

\subsection{Spectral Characterization for the Quadratic Case}
\label{Sec.3.5:Spectral Characterization for the Quadratic Case}
We can further characterize  the Dillon's property for $(n,m)$-functions in terms of their Walsh spectrum. In particular we have the following theorem.
\begin{theorem}
	\label{Theorem4}
	Let $m\geqslant n$ and let $F:\mathbb{F}_2^{n}\longrightarrow\mathbb{F}_2^{m}$ be a quadratic  function.
	Then $F$ satisfies the D-property if and only if
	\begin{equation*}
		\sum_{(u,v)\in \mathbb{F}_2^n\times\mathbb{F}_2^m}\mathcal W_{F+b}^3(u,v)>0
	\end{equation*} 
	for every $b\in\mathbb{F}_2^m\setminus\{0\}$.
\end{theorem}
\begin{proof}
    Thanks to Proposition \ref{prop: EA invariance}, we can assume without loss of generality that $F(0)=0$.
    For  sake of completeness, we  start by considering the sum in the statement  without the restriction $b\ne 0$, and see that in such case the conclusion is independent of the D-property. 
\begin{align*}
	&\sum_{(u,v)\in\ff n\times\ff m}\mathcal W_{F+b}^3(u,v)=\sum_{(u,v)\in\ff n\times\ff m}\sum_{x,y,z\in \mathbb{F}_2^n}\left(-1\right)^{v\cdot\left(F(x)+F(y)+F(z)+b\right)+u\cdot\left(x+y+z\right)}\\
	&=\sum_{x,y,z\in \mathbb{F}_2^n}\sum_{v\in\mathbb{F}_2^m}\left(-1\right)^{v\cdot\left(F(x)+F(y)+F(z)+b\right)}\sum_{u\in \mathbb{F}_2^n}\left(-1\right)^{u\cdot\left(x+y+z\right)}\\
 &=2^n\sum_{x,y\in \mathbb{F}_2^n}\sum_{v\in\mathbb{F}_2^m}\left(-1\right)^{v\cdot\left(F(x)+F(y)+F(x+y)+b\right)}\\
	&=2^{n+m}\cdot\left|\left\{(x,y)\in (\ff n)^2 : F(x)+F(y)+F(x+y)=b\right\}\right|\\
	&=2^{n+m}\Big(\left|\left\{(x,y)\in(\ff n\setminus\{0\})^2:x\ne y, F(x)+F(y)+F(x+y)=b\right\}\right|+(3\cdot2^n-2)\delta_0(b)\Big).
\end{align*}

Since $F(0)=0$, the  above sum corresponds to
\begin{equation}    \label{sumcubes_proofFormula}
 2^{n+m}\Big(\left|\left\{(x,y)\in(\ff n\setminus\{0\})^2:x\ne y, D^2_{x,y}F(0)=b\right\}\right|+(3\cdot2^n-2)\delta_0(b)\Big).
\end{equation}
Now we have two cases depending whether $b=0$ or not.
\begin{itemize}
	\item Case $b=0$.
	 Equation \eqref{sumcubes_proofFormula} becomes
	\begin{align*}
		\sum_{(u,v)\in\ff n\times\ff m}&\mathcal W_{F}^3(u,v)\\
  =&2^{n+m}\Big(\left|\left\{(x,y)\in(\ff n\setminus\{0\})^2 : x\ne y, D^2_{x,y}F(0)=0\right\}\right|+3\cdot2^n-2\Big).
	\end{align*}
 Notice that this quantity is positive  for any quadratic function,
	 regardless of  being a D-function.
	\item Case $b\ne 0$.
	 Equation \eqref{sumcubes_proofFormula} becomes
	\begin{equation*}
  \sum_{(u,v)\in\ff n\times\ff m}\mathcal W_{F+b}^3(u,v)=2^{n+m}\cdot\left|\left\{(x,y)\in (\ff n\setminus\{0\})^2 : x\ne y, D^2_{x,y}F(0)=b\right\}\right|.
	\end{equation*}
	From Proposition \ref{Prop1 Taniguchi},  $F$  satisfies the D-property if and only if
	\begin{equation*}
		\left|\left\{(x,y)\in (\ff n\setminus\{0\})^2 : x\ne y, D^2_{x,y}F(0)=b\right\}\right|>0
	\end{equation*}
	for every $b\in\mathbb{F}_2^m\setminus\{0\}$. This concludes the proof. \qedhere
\end{itemize}
\end{proof}

\begin{remark}
    Notice that, if $F$ as in the above theorem is APN, then the case $b=0$ corresponds to $\sum_{(u,v)\in\ff n\times\ff m}\mathcal W_{F}^3(u,v)=2^{n+m}\left(3\cdot 2^n-2\right)$.
    Indeed, as $F$ is APN, we have that 
    $\left\{(x,y)\in (\ff n\setminus\{0\})^2 : x\ne y, D^2_{x,y}F(0)=0\right\}$ is the empty set.
\end{remark}

As for the fourth moment, we can obtain an alternative characterization from the above theorem if we introduce the following pseudo-Boolean function associated to the function $F$
\begin{align*}
	&\hspace{-1pc}H_{3,F}:\mathbb{F}_2^m\longrightarrow\mathbb{Z}\\
	&\hspace{2pc}v\longmapsto\sum_{u\in \mathbb{F}_2^n}\mathcal W_F^3(u,v).
\end{align*}

From Theorem \ref{Theorem4} we can express the D-property of quadratic function in terms of the Fourier-Hadamard transform of the pseudo-Boolean function $H_{3,F}$.
\begin{prop}
	Let $m\geqslant n$ and let $F:\mathbb{F}_2^{n}\longrightarrow\mathbb{F}_2^{m}$ be a quadratic  function.
	Then $F$ satisfies the D-property if and only if
	\begin{equation*}
		\mathcal{F}\left(H_{3,F}\right)(b)>0\hspace{2pc}\forall b\in\mathbb{F}_2^m.
	\end{equation*}
\end{prop}

\subsection{A Further Simplification for the Case of Quadratic Functions}
\label{Sec.3.6:A Further Simplification for the Case of Quadratic Functions}
Proposition \ref{Prop1 Taniguchi} tells us that it is sufficient to consider the value taken by the second order derivatives at the origin.
Thanks to the bilinearity of the second order derivative $D^2_{\alpha,\beta}F$ with respect to the directions $\alpha,\beta$, it is possible to obtain a further simplification for testing the Dillon's property in the case of quadratic functions. Indeed, for a quadratic $(n,m)$-function, it is enough to check the values of the second order derivatives at the origin with one of the two directions restricted to a hyperplane of $\mathbb{F}_2^n$.
\begin{theorem}
	\label{simplification checking D property for quadratic}
	Let $F:\mathbb{F}_2^n\rightarrow\mathbb{F}_2^m$ be a quadratic function, with $n\leqslant m$. Let $K\subset\mathbb{F}_2^n$ be a vector subspace of dimension $n-1$ over $\mathbb{F}_2$.\\ 
	Then $F$ satisfies the D-property if and only if
	\begin{equation*}
		\left\{D^2_{\alpha,\beta}F(0):\alpha\in K,\beta\in\mathbb{F}_2^{n}\right\}=\mathbb{F}_2^m.
	\end{equation*}
\end{theorem}
\begin{proof}
    One direction is easy to see; if $F$ is such that the equality above holds, then it trivially satisfies the D-property.
    
We show that also the converse is true.
Assume that $F$ is a quadratic function  satisfying the D-property and let $b\in\mathbb{F}_2^m$ arbitrary. We need to show that there exist $\gamma_K\in K$ and $\beta'\in\mathbb{F}_2^n$ such that
\begin{equation*}
	D^2_{\gamma_K,\beta'}F(0)=b.
\end{equation*}
Since $K\subset\mathbb{F}_2^n$ with $\dim_{\mathbb{F}_2}K=n-1$, we know that
\begin{equation*}
	\exists c\in\mathbb{F}_2^n\setminus K\text{ s.t. }\forall\alpha\in\mathbb{F}_2^n\hspace{.5pc}\exists!\hspace{.3pc}\alpha_K\in K, \delta_{\alpha}\in\mathbb{F}_2:\alpha=\alpha_K+\delta_{\alpha}c.
\end{equation*}
Now, as $F$ is a D-function, we know that there exist $\alpha,\beta\in\mathbb{F}_2^n$ such that $D^2_{\alpha,\beta}F(0)=b$.
Clearly if either $\delta_{\alpha}=0$ or $\delta_{\beta}=0$, i.e.\ $\alpha\in K$ or $\beta\in K$, we conclude. So assume that $\delta_{\alpha}=1$ and $\delta_{\beta}=1$, then
\begin{align*}
	b=D^2_{\alpha_K+c,\beta_K+c}F(0)&=D^2_{\alpha_K,\beta_K}F(0)+D^2_{\alpha_K,c}F(0)+D^2_{c,\beta_K}F(0)=\\
	&=D^2_{\alpha_K,\beta_K}F(0)+D^2_{\alpha_K+\beta_K,c}F(0)=\\
	&=D^2_{\alpha_K+\beta_K,\beta_K}F(0)+D^2_{\alpha_K+\beta_K,c}F(0)
\end{align*}
where in the last equality we used the general relation:
\begin{equation*}
	D^2_{\alpha,\beta}F(p)=D^2_{\alpha+\beta,\beta}F(p).
\end{equation*}
Now we can conclude using again linearity of the second order derivative with respect to the directions to obtain
\begin{equation*}
	b=D^2_{\alpha_K+\beta_K,\beta_K+c}F(0)
\end{equation*}
and since $\alpha_K+\beta_K\in K$, this concludes the proof.
\end{proof}

Thanks to the above theorem, we were able to test the Dillon's property for the Gold functions $F(X)=X^{2^i+1}$ defined over $\mathbb{F}_2^{n+1}$, with $\gcd(i,n+1)=1$ and restricted to the $n$-dimensional subspace of the trace-zero elements, for every odd dimension $n+1$ such that $17\leqslant n+1\leqslant25$. 
The same was performed for the function $X^3+Tr_{n+1}(X^9)$ defined over $\mathbb{F}_2^{n+1}$,  restricted to the $n$-dimensional subspace of the trace-zero elements, for every odd dimension $n+1$ such that $17\leqslant n+1\leqslant25$.
Thus, thanks to Theorem \ref{Theor n+1 odd taniguchi} we obtain the following extension of the family presented by Taniguchi in \cite{taniguchi}.
\begin{theorem}
Over $\FF{n+1}$
	 the Gold APN functions given by $F(X)=X^{2^i+1}$,  with $\gcd(i,n+1)=1$, and the APN function  $X^3+Tr_{n+1}(X^9)$ satisfy the D-property when restricted to the trace zero elements for every dimension $n+1=17s,19s,21s,23s,25s$ for any positive integer $s$.
\end{theorem}

\subsection{Characterization from the ANF of Quadratic Functions}
We can derive an alternative characterization of the Dillon's property for quadratic functions by means of their ANF representation.

\begin{theorem}\label{thm8}
	Consider a quadratic $(n,m)$-function $F$  
	\begin{equation*}
		F(x)=\sum_{i=1}^{n-1}\sum_{j=i+1}^na_{i,j}x_ix_j+\sum_{k=1}^na_kx_k+a_0,  
	\end{equation*}
	with $a_{i,j},a_k,a_0\in\ff m$.
 Then $F$ satisfies the D-property if and only if
	\begin{equation*}
		\bigcup_{J\subseteq\{1,\ldots,n-1\}}\left\langle \sum_{j\in J\setminus\{i\}}a_{i,j}:i=1,\ldots,n\right\rangle=\mathbb{F}_2^m,
	\end{equation*}
    where $\langle S\rangle$ denotes the subspace of the $\f2 $-linear spanned by the vectors in $S$. 
 
\end{theorem}
\begin{proof} Given a vector $\alpha\in\ff k$, with $\alpha^{(\ell)}$ we indicate its $\ell$-th coordinate. 
Given $F$ as in the hypothesis and $\alpha,\beta\in\mathbb{F}_2^n$, we have
\begin{align}
		 \label{second derivative quad}
		&D^2_{\alpha,\beta}F(x)=\sum_{i<j}\left(\alpha^{(i)}\beta^{(j)}+\alpha^{(j)}\beta^{(i)}\right)a_{i,j}=\nonumber\\
		&=\sum_{i <j}\alpha^{(i)}\beta^{(j)}a_{i,j}+\sum_{i >j}\alpha^{(i)}\beta^{(j)}a_{i,j}
		=\sum_{i\ne j}\alpha^{(i)}\beta^{(j)}a_{i,j},
	\end{align}
 where  we use the convention $a_{i,j}=a_{j,i}$.

	We know that the map
	\begin{align*}
		D^2F:&\mathbb{F}_2^n\times\mathbb{F}_2^n\longrightarrow\mathbb{F}_2^m\\
		&(\alpha,\beta)\longmapsto D^2_{\alpha,\beta}F(0)
	\end{align*}
	is $\mathbb{F}_2$-bilinear. Therefore its image is given by the following union
	\begin{equation*}
		\bigcup_{\alpha\in\mathbb{F}_2^n}\left\langle D^2_{\alpha,e_i}F(x):e_i\in\mathcal{B}\right\rangle
	\end{equation*}
	where $\mathcal{B}$ is an $\mathbb{F}_2$-basis of $\mathbb{F}_2^n$. 
	Hence $F$ is a D-function if and only if this union of subspaces covers all $\mathbb{F}_2^m$.\\
	Thanks to Theorem \ref{simplification checking D property for quadratic} we can actually say that $F$ is a D-function if and only if
	\begin{equation*}
		\bigcup_{\alpha\in K}\left\langle D^2_{\alpha,e_i}F(x):e_i\in\mathcal{B}\right\rangle=\mathbb{F}_2^m
	\end{equation*}	
	for some $(n-1)$-dimensional vector subspace $K$ of $\mathbb{F}_2^n$.
	If we choose $\mathcal{B}$ to be the canonical basis given by $e_i^{(j)}=\delta_{i,j}$ (Kronecker delta) and $K$ to be the space generated by $\{e_1,\ldots,e_{n-1}\}$, we obtain that $F$ is a D-function if and only if
	\begin{equation*}
		\bigcup_{J\subseteq\{1,\ldots,n-1\}}\left\langle D^2_{\sum_{j\in J}e_j,e_i}F(x):e_i\in\mathcal{B}\right\rangle=\bigcup_{J\subseteq\{1,\ldots,n-1\}}\left\langle \sum_{j\in J}D^2_{e_i,e_j}F(x):e_i\in\mathcal{B}\right\rangle=\mathbb{F}_2^m.
	\end{equation*}		
	Using the expression (\ref{second derivative quad}) we have
	\begin{equation*}
		D^2_{e_i,e_j}F(x)=\sum_{\mu\ne\nu}e_i^{(\mu)}e_j^{(\nu)}a_{\mu,\nu}=\sum_{\mu\ne\nu}\delta_i^{(\mu)}\delta_j^{(\nu)}a_{\mu,\nu}=\begin{cases}
		     a_{i,j} &\mbox{if }i\ne j,\\
            0 &\mbox{otherwise}.
		\end{cases}	\end{equation*}
	Hence the thesis follows.
\end{proof}

As corollary of the above result we can derive a characterization of the APN property for quadratic functions in terms of their ANF. 
\begin{defi}
	Let $n$ be a positive integer. We introduce the following notations:
	\begin{align*}
 [n]^2=\{(i,j) :  i,j\in\{1,\dots,n\}\},&&
		\Delta_{n}=\{(i,i):i\in \{1,\dots,n\}\},&&
		\overline\Delta_n=[n]^2\setminus\Delta_{n}.
	\end{align*}
	We say that a non-empty subset $I\subset\overline\Delta_n$ is ultra-transitive if the following property is satisfied:
 \begin{center}
     for any $i,j,k,\ell\in[n]$ such that $(i,j),(\ell,k)\in I$ and $i\ne k$ then we have $(i,k)\in I$.
 \end{center}
\end{defi}
\begin{remark}
An ultra-transitive subset $I$ of $\overline\Delta_n$ can be visualized as a subset of points in $\overline\Delta_n$, such that it contains all the vertices of any rectangle we can construct from the points of $I$, unless the point lays in the diagonal $\Delta_{n}$ (which  it is not contained in the set $\overline\Delta_n$).
\end{remark}

We will use ultra-transitive sets as basis for evaluating second-order derivatives. In particular we want to associate to every pair of distinct nonzero directions $\alpha,\beta$ in $\mathbb{F}_2^n$ an ultra-transitive subset $I$ such that
\begin{equation*}
	D^2_{\alpha,\beta}F(0)=\sum_{(i,j)\in I}a_{i,j}.
\end{equation*}
To make this precise we need to take into account the symmetry of  the second-order derivative with respect to the directions.
This symmetry allows to remove some points from ultra-transitive sets. We can drop the points of $I$ which are symmetric with respect to the diagonal $\Delta_{n}$ as their mutual contribution in the sum $\sum_{(i,j)\in I}a_{i,j}$ would vanish (we adopted the convention $a_{i,j}=a_{j,i}$). This is described by the following reduction map:
\begin{align*}
	&UT(n)\longrightarrow \mathcal{P}\left(\overline\Delta_n\right)\\
	&\hspace{2pc}I\longmapsto\tilde{I}
\end{align*}
where $UT(n)$ denotes the family of ultra-transitive subsets of $\overline\Delta_n$ and $\tilde{I}$ is the subset $I$ to which we remove the pairs of the family
\begin{equation*}
	\{\{(i,j),(j,i)\}\}_{i\ne j}
\end{equation*}
which are fully contained in $I$.

In Figure \ref{fig1} we depict an example of 
 ultra-transitive set $I$ and the resulting set $\tilde{I}$.

\begin{figure}[h]  
\centering 
  \begin{subfigure}[b]{0.4\linewidth}\centering
    \begin{tikzpicture}[scale=0.5]   
     \draw[step=1cm,gray,very thin] (0,0) grid (6,6);
		\foreach \x in {1,2,3,4,5,6}
		\draw (\x cm,1pt) -- (\x cm,-1pt) node[anchor=north] {$\x$};
		\foreach \y in {1,2,3,4,5,6}
		\draw (1pt,\y cm) -- (-1pt,\y cm) node[anchor=east] {$\y$};
		\draw[thick,color=red](0,0) -- (6,6);
		\node[color=red] at (7,6.5) {$\Delta_{6}$};
		\draw[black, very thick] (1,2) rectangle (4,4);
		\draw[black, very thick] (1,2) rectangle (6,4);
		\draw[black, very thick] (1,2) rectangle (4,6);
		\draw[black, very thick] (1,2) rectangle (6,6);
		\node[color=blue] at (6.8,4) {$I$};
		\node at (1,2) [circle,color=blue,,fill,inner sep=0.1pt]{A};
		\node at (1,4) [circle,color=blue,,fill,inner sep=0.1pt]{B};
		\node at (4,2) [circle,color=blue,,fill,inner sep=0.1pt]{C};
		\node at (4,6) [circle,color=blue,,fill,inner sep=0.1pt]{D};
		\node at (1,6) [circle,color=blue,,fill,inner sep=0.1pt]{E};
		\node at (6,2) [circle,color=blue,,fill,inner sep=0.1pt]{F};
		\node at (6,4) [circle,color=blue,,fill,inner sep=0.1pt]{G};
		\node at (4,4) [circle,color=red,,fill,inner sep=0.1pt]{H};
		\node at (6,6) [circle,color=red,,fill,inner sep=0.1pt]{I};
    \end{tikzpicture}%
    \caption{An ultra-transitive subset $I$.} \label{ex.ultra-transitive}  
  \end{subfigure}
\begin{subfigure}[b]{0.4\linewidth}
  \begin{tikzpicture}[scale=0.5]  
      \draw[step=1cm,gray,very thin] (0,0) grid (6,6);
		\foreach \x in {1,2,3,4,5,6}
		\draw (\x cm,1pt) -- (\x cm,-1pt) node[anchor=north] {$\x$};
		\foreach \y in {1,2,3,4,5,6}
		\draw (1pt,\y cm) -- (-1pt,\y cm) node[anchor=east] {$\y$};
		\draw[thick,color=red](0,0) -- (6,6);
		\node[color=red] at (7,6.5) {$\Delta_{6}$};
		\draw[black, very thick] (1,2) rectangle (4,4);
		\draw[black, very thick] (1,2) rectangle (6,4);
		\draw[black, very thick] (1,2) rectangle (4,6);
		\draw[black, very thick] (1,2) rectangle (6,6);
		\node[color=blue] at (6.8,4) {$\tilde{I}$};
		\node at (1,2) [circle,color=blue,,fill,inner sep=0.1pt]{A};
		\node at (1,4) [circle,color=blue,,fill,inner sep=0.1pt]{B};
		\node at (4,2) [circle,color=blue,,fill,inner sep=0.1pt]{C};
		\node at (1,6) [circle,color=blue,,fill,inner sep=0.1pt]{E};
		\node at (6,2) [circle,color=blue,,fill,inner sep=0.1pt]{F};
    	\node at (4,4) [circle,color=red,,fill,inner sep=0.1pt]{H};
		\node at (6,6) [circle,color=red,,fill,inner sep=0.1pt]{I};
\end{tikzpicture}
\caption{The reduced set $\tilde{I}$.} \label{fig2}  
\end{subfigure}
\caption{An example for $n=6$, the red dots are not part of $I$ and $\tilde I$.}\label{fig1}
\end{figure}
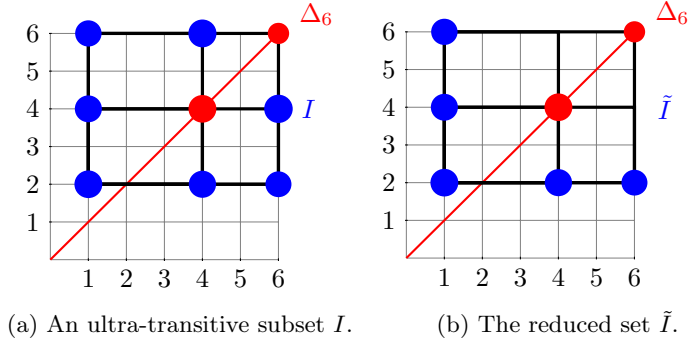

We can finally state and prove the following corollary which is a generalization of  \cite[Theorem 10]{Beth-Ding}. 
\begin{coro}
Consider a quadratic $(n,m)$-function as in Theorem \ref{thm8}.
	Then $F$ is APN if and only if for every ultra-transitive subset $I\subset\overline\Delta_n$ such that the reduced set $\tilde{I}$ is not empty, we have
	\begin{equation*}
		\sum_{(i,j)\in \tilde{I}}a_{i,j}\ne 0.
	\end{equation*}
\end{coro}
\begin{proof}
    We start proving that if $F$ is APN then for every ultra-transitive subset $I$ of $\overline\Delta_n$ the corresponding sum $\sum_{(i,j)\in \tilde{I}}a_{i,j}$ is nonzero. By contrapositive we suppose that there exists an ultra-transitive subset $I$ with $\tilde{I}\ne\emptyset$ and such that $\sum_{(i,j)\in \tilde{I}}a_{i,j}=0$, and we show that $F$ cannot be APN. So, let $I\subset \overline\Delta_n$ and $\tilde{I}$ be as described.
For $k=1,2$, let $\pi_k:\tilde{I}\rightarrow \{1,\ldots,n\}$ be the projection on the $k$-th component and, let $\alpha_k\in\mathbb{F}_2^n$ be the vector whose Hamming support coincide with $\pi_k(\tilde{I})$. So that,
\begin{equation}
	\label{conditions for I in APN quadratic proof}
	\begin{aligned}
		\alpha_{1}^{(i)}&=1\Leftrightarrow (i,\star)\in \tilde{I},\\
		\alpha_{2}^{(j)}&=1\Leftrightarrow (\star,j)\in \tilde{I},
	\end{aligned}
\end{equation}
where  $\star$ can represent any index in $\{1,\ldots,n\}$.
 Then, thanks to the bilinearity of the second order derivative and  to the ultra-transitive property of $I$, we can conclude that
\begin{equation*}
	D^2_{\alpha_1,\alpha_2}F(0)=\sum_{i\ne j}\alpha_1^{(i)}\alpha_2^{(j)}a_{i,j}
 =\sum_{i\ne j:\substack{\alpha_1^{(i)}=1\\\alpha_2^{(j)}=1}}a_{i,j}=\sum_{i\ne j:\substack{(i,\star)\in \tilde{I}\\(\star,j)\in\tilde{I}}}a_{i,j}=\sum_{(i,j)\in \tilde{I}}a_{i,j}=0.
\end{equation*}
Hence $F$ is not APN.

Conversely let us prove that, if for every ultra-transitive subset $I$ of $\overline\Delta_n$ such that $\tilde{I}\ne\emptyset$ the sum $\sum_{(i,j)\in \tilde{I}}a_{i,j}\ne 0$, then $F$ is APN. Again by contrapositive we suppose that $F$ is not APN, and we show that there exists an ultra-transitive subset $I\subset\overline\Delta_n$ such that $\tilde{I}\ne\emptyset$ and $\sum_{(i,j)\in I}a_{i,j}=0$. Since $F$ is not APN there exists two distinct nonzero vectors $\alpha_1,\alpha_2\in\mathbb{F}_2^n$ such that
\begin{equation*}
	D^2_{\alpha_1,\alpha_2}F(0)=0.
\end{equation*}
If we let $I$ to be the Cartesian product of the Hamming supports of $\alpha_1$ and $\alpha_2$ respectively, to which we remove the diagonal $\Delta_{n}$, i.e.
\begin{equation*}
	I=\left(Supp(\alpha)\times Supp(\beta)\right)\setminus\Delta_{n}
\end{equation*}
which means that:
\begin{equation*}
	(i,j)\in I \Leftrightarrow\begin{cases}
		\alpha_1^{(i)}&=1\\
		\alpha_2^{(j)}&=1
	\end{cases}\hspace{1pc}i\ne j.
\end{equation*}
We clearly have that $I$ is ultra-transitive and furthermore
\begin{equation*}
	\sum_{(i,j)\in \tilde{I}}a_{i,j}=\sum_{(i,j)\in I}a_{i,j}=\sum_{i\ne j:\substack{\alpha_1^{(i)}=1\\\alpha_2^{(j)}=1}}a_{i,j}=\sum_{i\ne j}\alpha_1^{(i)}\alpha_2^{(j)}a_{i,j}=D^2_{\alpha_1,\alpha_2}F(0)=0.
\end{equation*}
This concludes the proof.
\end{proof}

\section{A note on the case of Plateaued Functions}
The Dillon's property can be studied also for functions which are not strictly APN.
Indeed, Proposition \ref{Prop3 Taniguchi} tells us that  interesting conclusions can be drawn from the fact that a function has the D-property regardless of being APN.
For this reason we dedicate this last section to an important class of $(n,m)$-functions, namely plateaued functions, and we determine conditions under which they satisfy the Dillon's property.
\begin{theorem}\label{thm:plateaued}
	Let $F$ be a plateaued $(n,m)$-function (not necessarily with single amplitude). Then $F$ is a D-function if and only if for any $w\in\ff m$ we have
	\begin{equation*}
		2^{2n}+\sum_{v\in\mathbb{F}_2^m\setminus\{0\}}\lambda_v^2(-1)^{v\cdot w}>0,
	\end{equation*}
	where $\lambda_v$ is the amplitude of the plateaued component $F_v(x)=v\cdot F(x)$.
	\end{theorem}
\begin{proof}
   First, we define the pseudo-Boolean function
\begin{align*}
	\Lambda:&\mathbb{F}_2^m\longrightarrow\mathbb{R}\\
	&v\longmapsto\sum_{a,b\in\mathbb{F}_2^n}(-1)^{D^2_{a,b}F_v(x)}.
\end{align*}
This pseudo-Boolean function is well defined and independent on $x$ because $F$ is plateaued, see \cite{CP}. Furthermore, as we can deduce from \cite{CP}, we can express it as
\begin{equation}
	\label{Lambda function plateaued}
	\Lambda(v)=\begin{cases}
		\lambda_v^2&\text{ if }v\ne 0\\
		2^{2n}&\text{ if }v=0.
	\end{cases}
\end{equation}
Its Fourier-Hadamard transform gives
	\begin{align*}
	\mathcal{F}(\Lambda)(w)&=\sum_{v\in\mathbb{F}_2^m}\left(\sum_{a,b\in\mathbb{F}_2^n}(-1)^{D^2_{a,b}F_v(x)}\right)(-1)^{v\cdot w}=\sum_{v\in\mathbb{F}_2^m}\left(\sum_{a,b\in\mathbb{F}_2^n}(-1)^{v\cdot(D^2_{a,b}F(x)+w)}\right)\\
	&=2^m\cdot\left|\left\{(a,b)\in\mathbb{F}_2^n\times\mathbb{F}_2^n:D^2_{a,b}F(x)=w\right\}\right|.
\end{align*}
But from the expression \eqref{Lambda function plateaued} we can compute
\begin{equation*}
	\mathcal{F}(\Lambda)(w)=2^{2n}+\sum_{v\in\mathbb{F}_2^m\setminus\{0\}}\lambda_v^2(-1)^{v\cdot w}.
\end{equation*}
Hence
\begin{equation*}
	 \left|\left\{(a,b)\in\mathbb{F}_2^n\times\mathbb{F}_2^n:D^2_{a,b}F(x)=w\right\}\right|=\frac{1}{2^m}\left(2^{2n}+\sum_{v\in\mathbb{F}_2^m\setminus\{0\}}\lambda_v^2(-1)^{v\cdot w}\right).
\end{equation*}
Now the thesis easily follows from observing that $F$ is a D-function if and only if 
\begin{equation*}
	 \left|\left\{(a,b)\in\mathbb{F}_2^n\times\mathbb{F}_2^n:D^2_{a,b}F(x)=w\right\}\right|>0\hspace{1pc}\forall w\in\mathbb{F}_2^m.
\end{equation*}
\end{proof}

From the proof of Theorem \ref{thm:plateaued}, we have that the cardinality $|\{(a,b)\in\mathbb{F}_2^n\times\mathbb{F}_2^n:D^2_{a,b}F(x)=w\}|$ is independent from the element $x$. That is, if $|\{(a,b)\in\mathbb{F}_2^n\times\mathbb{F}_2^n:D^2_{a,b}F(x)=w\}|>0$ for some $x\in\ff n$, then $|\{(a,b)\in\mathbb{F}_2^n\times\mathbb{F}_2^n:D^2_{a,b}F(y)=w\}|>0$ for any $y\in\ff n$, and in particular for $y=0$. Thus, we obtain the following result.

\begin{coro}
    Let $F$ be a plateaued $(n,m)$-function (not necessarily with single amplitude), $F(0)=0$. Then $F$ is a D-function if and only if $\{ F(x)+F(y)+F(x+y): (x,y)\in\mathbb{F}_2^n\times\mathbb{F}_2^n\}=\ff m$.
\end{coro}

We recall  a known result on CCZ-equivalence of D-functions.

\begin{prop}[\cite{taniguchi}]
    Let $F,G : \ff n \to\ff m$ be CCZ-equivalent functions. If $F$ satisfies
$\{ F(x)+F(y)+F(x+y): (x,y)\in\mathbb{F}_2^n\times\mathbb{F}_2^n\}=\ff m$, then $G$ satisfies the D-property.
\end{prop}

Therefore, as for the case of quadratic functions we have that the D-property is CCZ-invariant for plateaued functions.

\begin{coro}
Let $F,G : \ff n \to\ff m$ be CCZ-equivalent functions. If $F$ is a plataued function that satisfies the D-property, then $G$ satisfies the D-property.
\end{coro}

\subsection{On Strongly Plateaued Functions}
We consider now strongly plateaued functions, a particular case of plateaued maps.
To introduce these functions, we need to recall the following definitions.
An $(n,1)$-function is \emph{bent} if any (nonzero) derivative is balanced or equivalently its nonlinearity equals $2^{n-1}-2^{\frac{n}{2}-1}$. Clearly, bent functions exist only for $n$ even.
Instead, an $(n,1)$-function is called \emph{partially-bent} if any derivative is either balanced or constant (other equivalent definitions can be found in \cite{carlet}).
Note that any quadratic $(n,1)$-function is partially-bent.

Finally, we say that an $(n,m)$-function $F$ is \emph{strongly plateaued} if 
 for any $v\in\ff m\setminus\{0\}$ the $(n,1)$-function $F_v$ is partially-bent.

 To proceed with the study of these functions, we introduce some useful notation.
The set of \emph{linear structures} of $F_v$ is  $V(v)=\{a\in\ff n : \deg(D_aF_v)=0\}$, which is an $\ff{}$-subspace of $\ff n$, and we denote its dimension with $\ell_v=\dim(V(v))$. 
Moreover, given $a\in\ff n$ we set $\Delta_a=\{ v\in\ff{m} : \deg(D_aF_v)=0\}$ the set of components for which $a$ is a linear structure.

We restrict then to $m=n+1$ with $n$ even, and we obtain the following result for strongly plateaued APN functions.
\begin{theorem}
\label{thm: strongly pl}
    For $n$ even, let $F'$ be a strongly plateaued APN $(n+1,n+1)$-function.
    Set $F=F'_{|T}$ its restriction to any hyperplane of $\ff{n+1}$, so $F$ is a strongly plateaued APN $(n,n+1)$-function.
    Then $F$ is a D-function if and only if for any $w\in\ff {n+1}$ we have 
    $$|\Omega_w|>\frac{2^n-1}{3},$$ where
    $\Omega_w=\{ v\in\ff{n+1}\setminus\{0\} : F_v\mbox{ not bent and } v\cdot w=0\}$.
 \end{theorem}
\begin{proof}
    From  \cite[Theorem 8]{Nyb95}, we have 
   \begin{equation}\label{eq:l_v}
  \sum_{v\in\ff {n+1}\setminus\{0\}}(2^{\ell_v}-1)=3\cdot(2^n-1),
\end{equation}
and from \cite[Lemma 7]{Nyb95} we have  that for any nonzero $a\in\ff n$, $|\Delta_a|=4$.
Moreover, we know that $\ell_v$ is even for any nonzero $v\in\ff {n+1}$.
We divide the nonzero components of $F$ in the following two sets:\begin{align*}
    \mathcal B=&\{v\in\ff{n+1} : D_aF_v \mbox{ is balanced } \forall a\in\ff n\setminus\{0\} \},\\
    \mathcal{NB}=&\{v\in\ff{n+1} : v\not\in \mathcal B\cup\{0\} \}.
\end{align*}
The set $\mathcal B$ coincides with the set of bent components of $F$. Clearly, if $v\in \mathcal B$ then $\ell_v=0$.
We have therefore the following relation. \begin{align*}
 3\cdot(2^n-1)=&   \sum_{v\in\ff {n+1}\setminus\{0\}}(2^{\ell_v}-1)=\sum_{v\in \mathcal{NB}}(2^{\ell_v}-1)\\
 \ge&|\mathcal{NB}|\cdot(2^2-1)=3\cdot(2^{n+1}-1-|\mathcal B|).
\end{align*}
Hence we deduce $|\mathcal{NB}|\le 2^{n}-1$ ($|\mathcal B|\ge 2^{n+1}-2^n = 2^n$).
On the other side, we have that $|\mathcal{NB}|\ge2^n-1$.
This is deduced from   \cite[Theorem 10]{Nyb95}.
Indeed, each (nonzero) component of $F'$ has exactly one nonzero linear structure, with different linear structures for different components. Therefore $2^n-1$ nonzero linear structures belong to the hyperplane $T$, and the corresponding components are then non-bent also for the restricted map $F$.
So  $F$  must have exactly $|\mathcal B|=2^n$ bent components and $|\mathcal{NB}|=2^n-1$ non-bent components, from which we deduce, using Equation \eqref{eq:l_v}, that $\ell_v=2$ for any $v\in \mathcal{NB}$.

We consider then the relation studied in Theorem \ref{thm:plateaued} adapted to this setting,
\begin{equation}\label{eq:thm10}
    2^{2n}+\sum_{v\in\ff{n+1}\setminus\{0\}}\lambda_v^2(-1)^{v\cdot w},
\end{equation} for $w\in\ff{n+1}$.
Recall that $\lambda_v^2=2^{n+\ell_v}$, see e.g.\ \cite{Nyb95}, and set $\Omega_w=\{v\in \mathcal{NB} : v\cdot w=0\}$.
Note that for $w=0$ we have $\Omega_0=\mathcal{NB}$ and $|\Omega_0|>(2^n-1)/3$ for any $F$ as in the hypothesis. So we can restrict to the case $w\ne0$.

We can write relation \eqref{eq:thm10} as the following.
\begin{align*}
    \eqref{eq:thm10}=&  2^{2n}+\sum_{v\in\mathbb{F}_2^{n+1}\setminus\{0\}}2^{n+\ell_v}(-1)^{v\cdot w}=2^n\left(2^{n}+\sum_{v\in\mathbb{F}_2^{n+1}\setminus\{0\}}2^{\ell_v}(-1)^{v\cdot w}
    \right)\\
    =&2^n\left(2^{n}+\sum_{v\in\mathcal B}2^{\ell_v}(-1)^{v\cdot w}+\sum_{v\in \mathcal{NB}}2^{\ell_v}(-1)^{v\cdot w}
    \right)\\
    =&2^n\left(2^{n}+\sum_{v\in\mathcal B}(-1)^{v\cdot w}+4\sum_{v\in \mathcal{NB}}(-1)^{v\cdot w}
    \right)=2^n\left(2^{n}+\sum_{v\in\mathbb{F}_2^{n+1}\setminus\{0\}}(-1)^{v\cdot w}+3\sum_{v\in \mathcal{NB}}(-1)^{v\cdot w} \right)\\
    =&2^n\left(2^{n}-1+3\sum_{v\in \mathcal{NB}}(-1)^{v\cdot w} \right)=2^n\left(2^{n}-1+3|\Omega_w|-3(2^n-1-|\Omega_w|)\right)\\
    =&2^n(6\cdot|\Omega_w|-2(2^n-1))=2^{n+1}(3\cdot|\Omega_w|-(2^n-1)).
\end{align*}

From the proof of  Theorem \ref{thm:plateaued}, we know that any functions as in the hypothesis is such that the quantity in \eqref{eq:thm10} is non-negative.
Therefore for any $w\in\ff{n+1}$, $w\ne0$, $|\Omega_w|\ge\frac{2^n-1}{3}$.
Moreover, the function is a D-function if for any such $w$ $|\Omega_w|>\frac{2^n-1}{3}$.
\end{proof}

To conclude this part, we consider the case of Gold APN functions (in their univariate representation) $F'(X)=X^{2^i+1}\in\FF{n+1}$, with $n$ even and $i$ coprime with $n+1$.
We then restrict $F$ to the elements of $T_0=\{ x\in\FF{n+1} : Tr_{n+1}(x)=0\}$. 
It is possible to consider any hyperplane $H_\alpha=\{x\,:\, Tr_{n+1}(\alpha x)=0 \}$, since for a power function $G(X)=X^d$ we have $G_{|H_\beta}=G_{|H_\alpha}(\frac{\alpha}{\beta}X)=(\frac{\alpha}{\beta})^d G_{|H_\alpha}(X)$.

Therefore, $F:T_0\rightarrow\FF{n+1}$ ($F=F'_{|T_0}$) is a  strongly plateaued APN $(n,n+1)$-function.

Given $a\in T_0^*$, $D_aF'(aX)=a^{2^i+1}(X^{2^i}+X+1)$.
Note that, if $a$ is a linear structure for $Tr_{n+1}(vF'(X))$ ($v\ne0$), then it is a linear structure also for $Tr_{n+1}(vF(X))$.
In particular, $Tr_{n+1}(vD_aF'(X))=Tr_{n+1}(vD_aF'(aX))=Tr_{n+1}(va^{2^i+1})$.  
Therefore, $Tr_{n+1}(va^{2^i+1}(X^{2^i}+X))=Tr_{n+1}((va^{2^i+1}+v^{2^i}a^{2^i(2^i+1)})X^{2^i})$ is the constant zero if and only if $va^{2^i+1}+v^{2^i}a^{2^i(2^i+1)}=0$, that is, $(va^{2^i+1})^{2^i-1}=1$ or equivalently $va^{2^i+1}=1$. Recall indeed that $a,v\ne0$ and $\gcd(i,n+1)=1$.
So  $v=1/a^{2^i+1}$ and
 $$\mathcal{NB}=\{1/a^{2^i+1} : a\in T_0^*\}.$$
 The equality holds since $X^{2^i+1}$ permutes $\FF{n+1}$ and the non-bent components of $F$ are $2^n-1$.

Now $\Omega_w=\{v\in \mathcal{NB} : Tr_{n+1}( v w)=0\}=\{1/a^{2^i+1} : a\in T_0^* \mbox{ s.t. } Tr_{n+1}(w/a^{2^i+1})=0\}$.
We further study the elements in $\Omega_w$.
We have that $Tr_{n+1}(w/a^{2^i+1})=0$ implies the existence of $x\in\FF{n+1}$ such that $w/a^{2^i+1}=x^{2^i}+x$.

So we have $w=a^{2^i+1}(x^{2^i}+x)=D_aF'(ax)+F'(a)$.
 Therefore $w\in {\tt Im}(D_aF'+F'(a))$ and $\Omega_w=\{1/a^{2^i+1} : a\in T_0^* \mbox{ s.t. }  w\in {\tt Im}(D_aF'+F'(a))\}$,
 where with ${\tt Im}(G)$ we identify the set of images of the function $G$.

From the APN property of $F'$, $|{\tt Im}(D_aF'+F'(a))|=|{\tt Im}(D_aF')|=2^{n}$.
So
\begin{align*}
    \sum_{w\in\FF{n+1}^*}|\Omega_w|=&\sum_{a\in T_0^*} |{\tt Im}(D_aF'+F'(a))|=\sum_{a\in T_0^*}|{\tt Im}(D_aF')|=2^{n}(2^n-1).
\end{align*}

Suppose now that the studied function $F$ has the D-property.
This implies, from the above theorem, that
 $|\Omega_w|>(2^n-1)/3$ for every nonzero $w$, so $|\Omega_w|\ge(2^n-1)/3+1$. From this we have $2^{n}(2^n-1)= \sum_{w\in\FF{n+1}^*}|\Omega_w|\ge(2^{n+1}-1)[(2^n-1)/3+1]$, 
 from which we obtain $(2^{2n}-1)\ge 3(2^{n+1}-1)$.
Notice that for $n=2$ the inequality is not satisfied: $2^{2n}-1=15$ and $3\cdot(2^{n+1}-1)=21$.
This leads to the following result.
\begin{prop}
    For $n=2$ consider the APN function $F':\FF{n+1}\rightarrow\FF{n+1}$ $F'(X)=X^{2^i+1}$,  $i$ coprime with $n+1$ (i.e.\ $i=1,2$), restricted to a hyperplane $H_\alpha=\{x:\,Tr_{n+1}(\alpha x)=0\}$ ($\alpha\ne 0$), $F=F'_{|H_\alpha}$.
    Then $F$ is not a D-function.    
\end{prop}
    
\section{Connection with higher-order differentiability}

As we have seen in Section \ref{sec:2}, Dillon's property is related to the values of the second-order derivatives. 
Lai, in \cite{LAI}, introduced the idea of higher-order differential cryptanalysis which is a generalization of (first-order) differential cryptanalysis.
Later, in  \cite{K95} Knudsen used the higher-order differential cryptanalysis to attack ciphers that were proved to be secure against conventional differential attacks. 
For more details on (higher-order) differential attacks see \cite{BLMN,K95_2,TMM} and the reference therein. 

Recently, Tang \emph{et al.} \cite{TMM} investigate the second-order differential spectrum of the inverse functions. Following the notation in \cite{TMM}, we have
$$\mathcal N_F(\gamma,\eta,\omega)=|\{x\in\ff n : F(x)+F(x+\gamma)+F(x+\eta)+F(x+\gamma+\eta)=\omega\}|,$$
for $\gamma,\eta\in\ff n$ and $\omega\in\ff m$, while $F$ is an $(n,m)$-function. 
Therefore, we can define the second-order differential spectrum as the multiset
$$
[N_F(\gamma,\eta,\omega)\,:\,\gamma,\eta \in\ff{n}\setminus\{0\},\omega\in\ff{m} ].
$$

Hence, the D-property can be formulated in terms of this notion as follows.
\begin{theorem} Let $m\ge n$.
    An $(n,m)$-function $F$ has the D-property if and only if for any $\omega\in\ff m\setminus\{0\}$ there exist $\gamma,\eta\in\ff n$ such that  $\mathcal N_F(\gamma,\eta,\omega)>0$.
\end{theorem}
Moreover, notice that we can restrict the search to nonzero $\gamma,\eta\in\ff n$ with $\gamma\ne\eta$.

As a byproduct of the analysis performed in the previous sections, we obtain the following relation between the fourth moment of the Walsh transform of a  function and its second-order differential spectrum.
\begin{theorem}
    Let $m\ge n$ and $F:\ff n\rightarrow\ff m$.
    Then for every $b\in\ff m\setminus\{0\}$ we have
    $$\sum_{(u,v)\in\ff n\times\ff m}(-1)^{v\cdot b}\mathcal{W}_F^4(u,v)=2^{n+m}\sum_{\gamma,\eta\in\ff n}\mathcal N_F(\gamma,\eta,b).$$
\end{theorem}
\begin{proof}
    Following the proof of Theorem \ref{thm:4th moment} we have
    \begin{align*}
         \sum_{(u,v)\in\ff n\times\ff m}&(-1)^{v\cdot b}\mathcal W_{F}^4(u,v)\\
     &=2^{n+m}\cdot |\{(x,y,z)\in (\ff n)^3 : F(x)+F(y)+F(z)+F(x+y+z)=b\}|.
    \end{align*} 
    The proof follows from  observing that
    $|\{(x,y,z)\in (\ff n)^3 : F(x)+F(y)+F(z)+F(x+y+z)=b\}|=|\{(x,\gamma,\eta)\in (\ff n)^3 : F(x)+F(x+\gamma)+F(x+\eta)+F(x+\gamma+\eta)=b\}|=\sum_{\gamma,\eta\in\ff n}\mathcal N_F(\gamma,\eta,b)$.
\end{proof}
A similar result is derived for quadratic functions, connecting
 the third moment of their Walsh transform  and their second-order differential spectrum.

\begin{theorem} Let $m\ge n$.
    Given $F$ a quadratic  $(n,m)$-function such that $F(0)=0$,
    then
    $$\sum_{(u,v)\in\ff n\times\ff m}\mathcal W^3_{F+b}(u,v)=2^m\sum_{\gamma,\eta\in\ff n}\mathcal N_F(\gamma,\eta,b)$$
     for every $b\in\ff m$.
\end{theorem}
\begin{proof}
    From the proof of Theorem \ref{Theorem4},  for any  function $F$ we  have
     $$\sum_{(u,v)\in\ff n\times\ff m}\mathcal W^3_{F+b}(u,v)=2^{n+m}\cdot|\{ (\gamma,\eta)\in(\ff n)^2
 : F(\gamma)+F(\eta)+F(\gamma+\eta)=b\}|.$$
 To conclude, we notice that for $F$ quadratic and normalized, we obtain
 \begin{align*} \sum_{\gamma,\eta\in\ff n}\mathcal N_F(\gamma,\eta,b)=
2^{n}\cdot|\{ (\gamma,\eta)\in(\ff n)^2
 : F(\gamma)+F(\eta)+F(\gamma+\eta)=b\}|. 
 \end{align*}
\end{proof}

As well, for the case of plateaued functions, we obtain a relation between the amplitudes of their components and the second-order differential spectrum.

 \begin{theorem}
      For $m\ge n$,
     let $F$ be a plateaued $(n,m)$-function.
     Then $$\sum_{v\in\ff m\setminus\{0\}}\lambda_v^2(-1)^{v\cdot w}=2^{m-n}\cdot \sum_{(\gamma,\eta)\in(\ff n)^2}\mathcal N_F(\gamma,\eta,w)-2^{2n},$$
     for any $w\in\ff m$ and with $\lambda_v$ the amplitude of the plateaued component $F_v$.
 \end{theorem}
 \begin{proof}
     From the proof of  Theorem \ref{thm:plateaued}, for $F$ plateaued we have 
     \begin{align*}
         2^{2n}&+\sum_{v\in\ff m\setminus\{0\}}\lambda_v^2(-1)^{v\cdot w}\\
         &=2^m\cdot|\{(\gamma,\eta)\in(\ff n)^2
 : F(x)+F(\gamma+x)+F(\eta+x)+F(\gamma+\eta+x)=w\}|,
     \end{align*}
     where the value on the second line is independent from the choice of $x$.
     We conclude the proof by noticing that, since $F$ is plateaued,
     \begin{align*}
         &\sum_{(\gamma,\eta)\in(\ff n)^2}\mathcal N_F(\gamma,\eta,w)\\
         &\ =\sum_{(\gamma,\eta)\in(\ff n)^2}|\{x\in\ff n : F(x)+F(x+\gamma)+F(x+\eta)+F(x+\gamma+\eta)=w\}|\\
         &\ =\sum_{x\in\ff n}|\{(\gamma,\eta)\in(\ff n)^2 : F(x)+F(x+\gamma)+F(x+\eta)+F(x+\gamma+\eta)=w\}|\\
        &\ 
         =2^n\cdot|\{(\gamma,\eta)\in(\ff n)^2
 : F(x)+F(\gamma+x)+F(\eta+x)+F(\gamma+\eta+x)=w\}|.
     \end{align*}
 \end{proof}

\end{document}